%% file: main.tex
\documentclass[runningheads]{llncs}
\usepackage{graphicx}

\usepackage{amsmath}
\usepackage{amssymb}
\usepackage{mathpartir}
\usepackage[tableaux]{prooftrees}
\usepackage{pifont}
\usepackage{subcaption}
\usepackage{dsfont}
\usepackage{hyperref}
\usepackage{todonotes}
\usepackage{multirow}
\usepackage{bm}

\usepackage[shortlabels]{enumitem}
\newcommand{\subsumes}[0]{\mathrel{\ooalign{\hss$\subseteq$\hss\cr\kern0.7ex\raise0.45ex\hbox{\scalebox{0.65}{$\sigma$}}}}}

\def\orcidID#1{\href{http://orcid.org/#1}{\raisebox{-1.25pt}{\includegraphics{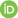}}}}

\usepackage{pifont}

\newcommand{\vampire}{\mbox{\textsc{Vampire}}}
\newcommand{\eprover}{\textsc{E}}
\newcommand{\spass}{\mbox{\textsc{Spass}}}
\newcommand{\setheo}{\mbox{\textsc{Setheo}}}
\newcommand{\leancop}{\mbox{\textsf{leanCoP}}}
\newcommand{\cadical}{\mbox{\textsc{CaDiCaL}}}
\newcommand{\zzz}{Z3}

\begin{document}
\newcommand{\titleText}{Spanning Matrices via Satisfiability Solving}
\title{\texorpdfstring{\titleText}{\titleText}}
\author{Clemens Eisenhofer\orcidID{0000-0003-0339-1580} \and
Michael Rawson\orcidID{0000-0001-7834-1567} \and
Laura Kov\'{a}cs\orcidID{0000-0002-8299-2714}}
\authorrunning{Eisenhofer et al.}
%
\institute{TU Wien, Austria\\
\email{\{clemens.eisenhofer,michael.rawson,laura.kovacs\}@tuwien.ac.at}}
\maketitle              
\begin{abstract}
We propose a new encoding of the first-order connection method as a Boolean satisfiability problem.
The encoding eschews tree-like presentations of the connection method in favour of matrices, as we show that tree-like calculi have a number of drawbacks in the context of satisfiability solving.
The matrix setting permits numerous global refinements of the basic connection calculus.
We also show that a suitably-refined calculus is a decision procedure for the Bernays-Sch\"onfinkel class.

\end{abstract}

\section{Introduction}
\input{introduction}

\section{Preliminaries}
\label{sec:preliminaries}
\input{preliminaries}

\section{Encoding Connection Tableaux}
\label{sec:tableau}
\input{tableaux}

\section{Encoding Matrices}
\label{sec:matrix}
\input{matrix}

\section{Iterative Deepening via Unsat Core Refinement}
\label{sec:unsatCore}
\input{unsatCore}

\section{Redundancy Elimination in SAT Solving}
\label{sec:optimisations}
\input{optimisations}

\section{Deciding Bernays-Sch\"onfinkel}
\label{sec:epr}
\input{epr}

\section{Related Work}
\input{related}

\section{Conclusion}
\input{conclusion}

\subsubsection{Acknowledgements.} {We acknowledge funding from the ERC Consolidator Grant ARTIST 101002685, the TU Wien SecInt Doctoral College, the FWF SFB project SpyCoDe F8504, and the WWTF ICT22-007 grant ForSmart.}
\bibliographystyle{splncs04}
\bibliography{main}

\end{document}

%% file: introduction.tex
Search strategies employed by \emph{automated theorem provers for first-order logics} can be divided into two broad classes~\cite{taxonomy}: \emph{ordering-based} and \emph{subgoal-reduction}. 
The first class, which contains saturation-based systems including \vampire~\cite{Vampire}, \eprover~\cite{eprover}, and \spass~\cite{spass}, work by continuously deducing new facts from an existing set of formulas.
The second class, containing systems such as \setheo~\cite{SETHEO} or \leancop~\cite{leanCoP}, work by manipulating a partial proof, backtracking as necessary.

The subgoal-reduction class has the disadvantage that redundant search space may be explored in duplicate unless care is taken to ``remember'' where one has been before.
Avoiding such cases by \emph{global} refinement is a subject of great interest among proponents of the subgoal-reduction approach to theorem proving~\cite{comparison-of-proof-methods}.
In general, such refinements can contain non-trivial propositional structure, such as the information ``if clauses $C$ and $D$ are in the current proof attempt, and the current substitution binds at least $x \mapsto t$ and $y \mapsto s$, we are in a dead-end and should backtrack''.

Backtracking mechanisms are routinely implemented in Boolean satisfiability (SAT) solvers~\cite{DBLP:conf/sat/BiereFW23,ipasir-up}.
Modern SAT solvers \emph{learn} relevant information as they go, and the most recent iterations even allow users to add constraints during the solver's search for a model, in response to the solver's current assignment.
When a solver cannot find a satisfying assignment, they can offer an \emph{explanation} in the form of an unsat core.
These features make satisfiability solvers an ideal vehicle for managing global information and thereby guiding proof search.

\emph{Here we are interested in the integration of SAT and subgoal-reduction}, focusing on Bibel's \emph{connection method}~\cite{Bibel}.
We directly encode search for connection proofs as a Boolean satisfiability problem, allowing the solver to dictate search decisions and responding by asserting constraints, such that when a satisfying assignment is reached, it represents a complete proof.
This approach can be applied to connection \emph{tableaux} (Section~\ref{sec:tableau}), but with some unfortunate properties, which motivates our encoding of the connection calculus in matrix form (Section~\ref{sec:matrix}).
Unsat cores are used to guide iterative deepening (Section~\ref{sec:unsatCore}), and furthermore the encoding allows many global refinements of the calculus that are usually not feasible within ordinary methods (Section~\ref{sec:optimisations}).

Our approach intends for the SAT solver to return a \emph{satisfying} assignment of our constraints, where the model represents a finished proof: matrix or tableau.
This contrasts with most other uses of SAT solvers in theorem proving in which ground \emph{unsatisfiability} is the aim~\cite{AVATAR}, often witnessing Herbrand-style refutation by instantiation of first-order clauses.

%% file: preliminaries.tex
We use the standard syntax and semantics of classical first-order logic~\cite{smullyan}.
Logical objects such as terms $t$ may be indexed: $t^i_j$.
We assume for the sake of a wider audience that the input problem has been negated and converted to conjunctive normal form (CNF) by a satisfiability-preserving transformation~\cite{handbook-ar-nf}, although the initial negation and CNF transformation are not strictly necessary~\cite{Bibel}.

\subsection{Satisfiability Solving}
We assume familiarity with Boolean satisfiability (SAT) solving~\cite{handbook-of-satisfiability} and satisfiability modulo theories (SMT)~\cite{DBLP:conf/jelia/Tinelli02}.
In addition to the basic decision procedure for Boolean formulas, many SAT solvers support solving \emph{under assumptions} and \emph{unsatisfiable cores}.
Solving under assumptions allows fixing some literals temporarily for the duration of a solving run: afterwards, the solver ``forgets'' them and their consequences.
If the solver detects that the problem is unsatisfiable under assumptions, it may extract a subset of the assumptions that were used to derive inconsistency: the so-called ``unsat core''.
Cores may not be \emph{minimal}, so inconsistency can be derived with a strict subset of the core.
Minimal cores can be generated at additional computational cost~\cite{minimal-unsat-cores-smt,minimal-unsat-cores}.

Some SAT and SMT solvers, including \cadical~\cite{ipasir-up} and \zzz~\cite{user-propagation}, allow the user to intervene \emph{during} search by a variety of means, often under the slogan ``user propagation''.
Such mechanisms allow employing a solver for tackling a broad class of problems efficiently.
For our purposes, we assume we can be notified when a SAT variable is assigned true or false, and respond by asserting additional constraints, potentially containing fresh SAT variables.
We write $J_1, \ldots, J_n \Vdash F$ to represent that we added (\emph{propagated}) the constraint $J_1 \wedge \ldots \wedge J_n \Rightarrow F$ to the solver given that the solver's current model satisfies all antecedents $J_1, \ldots, J_n$.
This feature allows us to avoid eagerly generating a very large set of all possible constraints and add only those parts of the encoding that are currently relevant. This kind of lazy generation is very desirable in our case.

\subsection{Connection Tableaux}
\label{sec:prelim_tableaux}
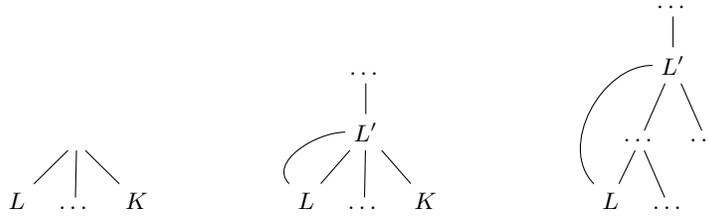
\begin{figure}[t]
    \centering
    \begin{forest}
        [,
        [$L$], [\ldots], [$K$]
        ]
    \end{forest}
    \hspace{.5in}
    \begin{forest}
        [\ldots,
        [$L'$, name=root
            [$L$, name=mate]
            [\ldots]
            [$K$]
        ]
        ]
        \draw (mate) to[out=north west, in=west] (root);
    \end{forest}
    \hspace{.5in}
    \begin{forest}
        [\ldots,
        [$L'$, name=ancestor
        [\ldots,
        [$L$, name=leaf], [\ldots]
        ], [\ldots]
        ]
        ]
        \draw (leaf) to[out=north west, in=west] (ancestor);
    \end{forest}
    \caption{Connection tableau rules, left-to-right: \emph{start}, \emph{extension}, and \emph{reduction}. In \emph{start} and \emph{extension}, $L \vee \ldots \vee K$ is a freshly-renamed copy of a clause from the input problem. In \emph{extension} and \emph{reduction}, $L$ is connected to $L'$ using $\sigma$.}
    \label{fig:tableau-rules}
\end{figure}
Connection tableaux are essentially clausal tableaux~\cite{tableauxHandbook} with the additional constraint that each clause added to a branch must have at least one literal \emph{connected} to the current leaf literal~\cite{handbook-ar-model-elimination}. Two literals are connected if they have the same atom but opposite polarity: they are dual.
Recall that in the first-order case, clauses in the tableau have their variables renamed apart from any other and a global substitution $\sigma$ is applied to the entire tableau in order to connect literals.
The connection tableaux calculus is not \emph{confluent} and therefore requires both backtracking and a fair enumeration of tableaux to retain completeness.

We say that two literals $L, K$ \emph{can} be connected and write $L \bowtie K$ if there exists some substitution $\rho$ such that $\rho(L)$ is connected to $\rho(K)$.
A subset of input clauses are considered potential roots of the tableau~\cite{handbook-ar-model-elimination}: we assume these \emph{start} clauses have been chosen appropriately.
Equality is not handled by the basic connection calculus, and it is either axiomatised~\cite{handbook-ar-paramodulation} or preprocessed away by some variation of Brand's modification~\cite{brand}.
We sometimes write $C^k$ to distinguish the $k$\textsuperscript{th} copy of the input clause $C$, indexing its variables $x^k$.

There are conventionally three operations manipulating connection tableaux, shown in Figure~\ref{fig:tableau-rules}.
\emph{Start} operations pick a start clause and add it at the root of the tableau.
\emph{Extension} operations add a clause below a leaf literal, connecting some literal in the clause with the leaf.
\emph{Reduction} operations connect a leaf literal with another literal on the path from the literal toward the tableau's root.
In general, all these operations must be backtracked over to achieve completeness, but the choice of leaf literal does not matter.

\subsection{The Connection Method, Matrices, and Spanning Connections}
Connection tableaux are closely related
to, or are an instance of, the connection method~\cite{Bibel}.
While connection calculi are a rich topic with many facets, we are primarily interested in the \emph{matrix}\footnote{readers may be familiar with other kinds of matrices: they are unrelated} representation~\cite{matings-in-matrices}: here, we consider matrices
in \emph{normal form} and therefore define a matrix to be a set of clauses.
As in Section~\ref{sec:prelim_tableaux}, clauses in a matrix are renamed apart, and a global substitution is applied.
For simplicity, we explicitly copy input clauses into the matrix and therefore present explicit rather than implicit amplification~\cite{matings-in-matrices}.
A \emph{path} through a matrix is a set containing exactly one literal from each clause in the matrix. A path is \emph{open} in case it does not contain at least one connected pair of literals, otherwise the path is \emph{closed}.
A matrix proof is \emph{finished} -- we have a \emph{spanning set of connections} -- when there does not exist an open path.

Further, a matrix is \emph{fully connected} with respect to a set of connections if each literal in the matrix is connected to at least one other literal of a different clause~\cite{using-matings-for-pruning}.
A matrix $M$ is \emph{minimal} if there is no proof using only a strict subset of $M$. 
Although a start clause must be in the matrix, there is no inherent tree structure in matrices, unlike connection tableaux.
To illustrate the two representations, consider the unsatisfiable set
\[
\begin{matrix}
\forall x\forall y.~&\lnot P(x)~ \vee \lnot &P(f(y))\\
\forall z.~&P(z)~ \vee &P(f(z))
\end{matrix}
\]
and compare matrix and tableau refutations thereof in Figure~\ref{fig:tabvsmat}. Figure~\ref{fig:openPath} shows a fully connected matrix that is not a proof because there is an open path.
\begin{figure}[t]
\begin{minipage}{.44\textwidth}
\begin{tikzpicture}[main/.style = {draw, circle}] 
\node (C31) at (0, 0) {$P(z^1)$}; 
\node (C32) at (0, -1) {$P(f(z^1))$}; 

\node (C11) at (2, 0) {$\lnot P(x^1)$}; 
\node (C12) at (2, -1) {$\lnot P(f(y^1))$}; 

\node (C21) at (4, 0) {$\lnot P(x^2)$}; 
\node (C22) at (4, -1) {$\lnot P(f(y^2))$}; 

\path (C11) edge[bend right=20] (C31);
\path (C12) edge[bend right=20] (C31);

\path (C21) edge[bend left=40] (C32);
\path (C22) edge[bend left=40] (C32);
\end{tikzpicture}
\end{minipage}
\begin{minipage}{.44\textwidth}
\centering
\begin{forest}
[
 [$P(z^1)$, name=pz
  [$\lnot P(x^1)$, name=px1]
  [$\lnot P(f(y^1))$, name=pfy1]
 ]
 [$P(f(z^1))$, name=pfz
  [$\lnot P(x^2)$, name=px2]
  [$\lnot P(f(y^2))$, name=pfy2]
 ]
]
\draw (px1) to[out=120, in=west] (pz);
\draw (pfy1) to[out=north, in=east] (pz);
\draw (px2) to[out=120, in=west] (pfz);
\draw (pfy2) to[out=north, in=east] (pfz);
\end{forest}
\end{minipage}
\caption{Matrix versus tableau proofs. Clauses are written vertically in matrices. Curved lines indicate connections. $\sigma$ is computed such that e.g. $\sigma(z^1) = f(y^1)$.}
\label{fig:tabvsmat}
\end{figure}
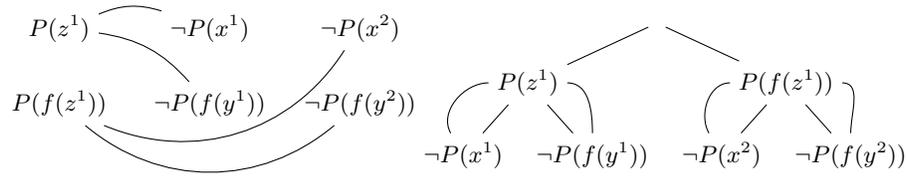
\begin{figure}
\begin{center}
\begin{minipage}{.47\textwidth}
\begin{tikzpicture}[main/.style = {draw, circle}] 
\node (C31) at (0, 0) {$P(z^1)$}; 
\node (C32) at (0, -1) {$\bm{P(f(z^1))}$}; 

\node (C11) at (2, 0) {$\lnot P(x^1)$}; 
\node (C12) at (2, -1) {$\bm{\lnot P(f(y^1))}$}; 

\node (C21) at (4, 0) {$\bm{\lnot P(x^2)}$}; 
\node (C22) at (4, -1) {$\lnot P(f(y^2))$}; 

\path (C11) edge[bend right=20] (C31);
\path (C12) edge[bend right=20] (C31);

\path (C21) edge[bend right=30] (C31);
\path (C22) edge[bend left=30] (C32);
\end{tikzpicture} 
\end{minipage}
\end{center}
\caption{Unfinished matrix proof. An open path is shown in \textbf{bold}.}
\label{fig:openPath}
\end{figure}
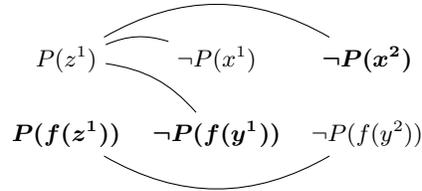

%% file: tableaux.tex
We first encode the search for closed connection tableaux in a SAT solver. 
A closed connection tableau is  explicitly constructed from a satisfying assignment.
We encode that a literal $L$ is part of the tableau at path $U$ using a SAT variable $\langle L; U\rangle$.
$L$ and $U$ have no inherent meaning to the solver and are used only to determine the corresponding variable. For example, $\lnot P(x^2)$ in Figure~\ref{fig:tabvsmat} is represented by a variable $\langle \lnot P(x^2); \{ P(f(z^1)) \}\rangle$. 
The substitution $\sigma$ and unification of connected literals 
are handled with another family of variables we discuss later. 

\paragraph{\bf Connection tableau rules as SAT.} We begin by asserting that at least one $S$ of the start clauses must be present in the tableau. Therefore, all literals $L \in S$ must be in the tableau at the root:
\begin{equation}
\label{eq:start}
\bigvee_S \bigwedge_{L \in S} \langle L; \emptyset \rangle
\end{equation}
The SAT solver is free to choose any start clause $S$, but all literals in the chosen  $S$ must be present at the root of the tableau.
As the solver assigns variables $\langle L; U \rangle$ true we respond by propagating additional requirements.
We  demand that each literal has either an \emph{extension} $E_{C, K}$ or a \emph{reduction} $R_{K}$ applied, in order to close the corresponding branch in the final tableau:
\begin{equation}
\label{eq:tableauxEncoding}
\langle L; U \rangle \Vdash \bigvee_{C, K} E_{C, K} \vee \bigvee_{K \in U} R_{K}
\end{equation}
Each formula $E_{C, K}$ represents applying an extension operation at $L$ using a fresh copy of a clause $C$ containing a literal $K \bowtie L$, which yields
\begin{equation}
\label{eq:ext}
E_{C, K} := \left[\langle L \sim K\rangle \wedge \bigwedge_{\substack{K' \in C\\K' \neq K}} \langle K'; \{ L \} \cup U \rangle\right]
\end{equation}
i.e. that if an extension step $E_{C, K}$ is taken, $L$ and $K$ are connected and other literals $K' \in C$ must be in the tableau with path $\{ L \} \cup U$.
We write $\langle L \sim K \rangle$ for the SAT variable representing that $L$ and $K$ are connected modulo $\sigma$. Similarly,
\begin{equation}
\label{eq:red}
R_K := \langle L \sim K\rangle
\end{equation}
where $K \bowtie L$ is on the path $U$.
The possible steps $E_{C, K}$ and $R_{K}$ are computed based only on the \emph{possible} connection relation $\bowtie$: the current substitution $\sigma$ is ignored, as it may change with solver decisions elsewhere in the tableau.
Iterative deepening may be applied as usual~\cite{leanCoP}, perhaps by offering no $E_{C, K}$ alternatives if the path length $|U|$ exceeds a depth limit.
For example, when the solver decides that $\lnot P(x^2)$ is present in Figure~\ref{fig:tabvsmat} we propagate $\langle \lnot P(x^2); \{ P(f(z^1)) \}\rangle \Vdash$

\[ \begin{array}{lllrrrrr}
    &[ &\langle P(f(z^2)); \{ P(f(z^1)), \lnot P(x^2) \} \rangle& &\wedge& &\langle\lnot P(x^2) \sim P(z^2)\rangle &]\vee \vphantom{A}\\
    &[ &\langle P(z^2); \{ P(f(z^1)), \lnot P(x^2) \} \rangle& &\wedge& &\langle\lnot P(x^2) \sim P(f(z^2))\rangle &]\vee \vphantom{A}\\
    &&&&&&\langle\lnot P(x^2) \sim P(f(z^1))\rangle
\end{array} \]

\paragraph{{\bf Unification Constraints.}} 
Variables $\langle L \sim K\rangle$ constrain $\sigma$ such that $\sigma(L)$ is connected $\sigma(K)$.
When the SAT solver assigns such a variable, we  check whether this is consistent with the existing set of constraints.
This can be done by applying a unification algorithm, perhaps using an efficient data structure such as the \emph{variable trail}~\cite{handbook-ar-model-elimination} to handle backtracking.
We note in passing that algebraic datatype solvers~\cite{datatypes} implement a similar decision procedure.
If the constraints are not satisfiable, we produce a \emph{conflict clause} containing the reasons as Boolean assignments.
For example, if we have $\langle L \sim K\rangle$, $\langle J \sim K\rangle$ and $\langle L \sim J\rangle$, but $\langle L \sim K\rangle \wedge \langle L \sim J\rangle$ is already unsatisfiable, we add the conflict
\begin{equation}
\lnot\langle L \sim K\rangle \vee \lnot\langle L \sim J\rangle
\end{equation}
causing the solver to backtrack.
This approach also allows a uniform treatment of refinements such as \emph{regularity} based on \emph{disequation constraints}~\cite{handbook-ar-model-elimination}.

\paragraph{\bf SAT Encoding of Closed Connection Tableaux.} We now have all the ingredients for our SAT encoding, which we denote by $\mathcal{E}_T$. By asserting that (i) a start clause must be present~\eqref{eq:start}, (ii) each literal in the tableau must have a reduction or extension rule applied to it~\eqref{eq:tableauxEncoding} and (iii) connections must have a consistent unifier, enforced by unification constraints, our encoding $\mathcal{E}_T$ is complete. Each propositional model of $\mathcal{E}_T$ represents a closed connection tableau.

\paragraph{\bf Pathological Behaviour.}
Our SAT encoding $\mathcal{E}_T$, while simple, has severe drawbacks. 
The most important is that extension adds a \emph{fresh instance} from the clause set to the tableau and so the number of different SAT variables $\langle L; U \rangle$ grows rapidly.
In turn, this means the resulting SAT problem has only limited propositional structure between variables that the solver can exploit.
Search tends to degrade towards the kind of exhaustive enumeration that a system such as \leancop{}~\cite{leanCoP} implements, but with the added overhead of a SAT solver.

%% file: matrix.tex
To avoid the problems of $\mathcal{E}_T$, we encode matrix proofs. We denote our matrix-based encoding $\mathcal{E}_M$. 
Most search routines for spanning set of connection presented in literature~\cite{Bibel,leanCoP} restrict connections such that the matrix form simulates one or more connection tableaux, but this is not strictly necessary~\cite{matrix-based-constructive}.
In our $\mathcal{E}_M$ encoding, we allow arbitrary connections between clauses present in the matrix.
A single proof in the matrix representation can correspond to numerous proofs in the tableau form~\cite{using-matings-for-pruning}.
In any event, our new representation $\mathcal{E}_M$ produces a combinatorial problem of finding connections between a set of clauses, which we argue is much more suitable for SAT solvers than $\mathcal{E}_T$.

\subsection{Encoding Overview}
We find a matrix with a given \emph{resource limit} and span it in two steps:
\begin{enumerate}
    \item We encode constraints for a fully-connected matrix (Section~\ref{sec:fully-connected}).
    \item We constrain that the result has a set of spanning connections (Section~\ref{sec:check}).
\end{enumerate}
We use the following result to motivate our encoding.
\begin{theorem}[Fully Connected Matrix]
\label{thm:connected}
Suppose $M$ is minimal and has a spanning set of connections. Then $M$ is fully connected.
\end{theorem}
\begin{proof}
First note that this is similar but not quite identical to Proposition~1 in Letz's work on matings pruning~\cite{using-matings-for-pruning}.
Suppose towards contradiction, there is a literal $L \in C \in M$ that is not connected to any other $K$.
Now consider the rest of the matrix $M' = M \setminus {C}$.
Since $M$ is minimal, there is an open path $U$ through $M'$, otherwise we could span $M'$.
Therefore, $U \cup \{L\}$ is an open path for $M$.
\qed
\end{proof}
Theorem~\ref{thm:connected} allows us to  restrict our work to fully-connected matrices. This restriction is a good approximation, as few fully-connected matrices are not spanning.
In the following, we use SAT variables of the form $S_C$ to denote that clause $C$ appears in the matrix, sometimes superscripted $S_C^k$ to indicate selecting $C^k$, the $k^{th}$ copy of $C$. We call these $S_C$ \emph{selectors} and call $C$ \emph{selected} if $S_C$ is assigned true.
At least one of the start clauses $C$ must be selected, cf.~\eqref{eq:start}: 
\begin{equation}\label{eq:start:matrix}
\bigvee S^1_C. 
\end{equation}
In Section~\ref{sec:tableau} we apply iterative deepening on the maximum length of a branch. This kind of resource limit cannot be applied here, where there is no obvious notion of \emph{branch}, so we must come up with alternatives. We first apply iterative deepening on the number of clause $d$ in the matrix. We can see immediately that we need only introduce at most $d$ selectors for each clause. As we always refer to these copies, the solver can more easily learn propositional structure than with $\mathcal{E}_T$. We discuss a further enhanced encoding later in Section~\ref{sec:unsatCore}.

\subsection{Fully Connected Matrices}
\label{sec:fully-connected}
By Theorem~\ref{thm:connected} we may constrain that each literal in the matrix must connect to at least one other literal.
Similarly to \eqref{eq:tableauxEncoding}, we respond to a selection $S_C$ by propagating that each literal must be connected to some other literal in another clause in the matrix by enumerating all possible connections.
This other clause could be selected  or require selection, but there is no distinction between extension and reduction.
Suppose $C$ is selected. For each $L \in C$, we propagate
\begin{equation}
\label{eq:matrixConnect}
    S_C \Vdash~\bigvee_D~\bigvee_{1 \le k \le d}\bigvee_{K \in D^k} S^k_D \wedge \langle L \sim K\rangle
\end{equation}
where $K \bowtie L$ is a literal in the input clause $D$ we connect to, and $k$ indicates which copy $D^k$ of that clause is used.

To enforce that there are at most $d$ clauses selected for the matrix, there are several possible options.
We suggest using pseudo-Boolean constraints~\cite{DBLP:journals/tcad/ChaiK05} or a direct encoding~\cite{DBLP:conf/cp/BailleuxB03,handbook-of-satisfiability,DBLP:conf/cp/Sinz05} to constrain that ``there are no more than $d$ selector variables assigned''.
We can strengthen this to \emph{exactly} $d$ as we apply iterative deepening, so the less-than-$d$ case was encountered already.

\subsection{Spanning Sets of Connections}
\label{sec:check}
Once we have a fully connected matrix, we check for open paths.
If there are none, we are done and can use the resulting SAT model to output a proof consisting of the matrix and the spanning set of connections.
Suppose instead there is an open path $U$ through the matrix $M$.
At least two literals along $U$ must connect in order to span $M$.
Let $\bar{S}$ be the set of selectors assigned true.
Propagating
\begin{equation}
\label{eq:openpath1}
\bar{S} \Vdash \bigvee_{\{L, K\} \subseteq U} \langle L \sim K\rangle
\end{equation}
forces the solver to ``fix'' $M$, likely via backtracking, by requiring that $U$ is not an open path.

\subsection{Correctness and Complexity of Matrix Encodings}
Our encoding $\mathcal{E}_M$ consists of~\eqref{eq:start:matrix},~\eqref{eq:matrixConnect},~\eqref{eq:openpath1}, and constraints for the depth limit.
It models search for a matrix with a spanning set of connections.
We show soundness, completeness, and termination for a given size $d$  in $\mathcal{E}_M$, and describe the respective complexity class of $\mathcal{E}_M$.

\begin{theorem}[Soundness]
\label{thm:sound}
A propositional model of $\mathcal{E}_M$ represents a matrix with a spanning set of connections.
\end{theorem}
\begin{proof}
Whenever the SAT solver finds a propositional model, we first check that it represents a proof, adding constraints if not (Section~\ref{sec:check}).
\qed
\end{proof}

\begin{theorem}[Completeness]
\label{thm:complete}
\label{thm:rectangular}
If a matrix $M$ together with a spanning set of connections exists, there is a propositional model of $\mathcal{E}_M$ at depth $d = |M|$.
\end{theorem}
\begin{proof}
$M$ can be represented by setting $S^k_C$ true iff there are at least $k$ copies of $C$ in $M$.
The spanning set of connections is represented by setting $L \sim K$ iff $L$ is connected to $K$ in the proof.
This model of $\mathcal{E}_M$ and all its submodels are consistent modulo the semantics of $\sim$ and all possible instances of~\eqref{eq:matrixConnect}. Furthermore, the final model satisfies the depth constraints and contains at least one start clause.
We do not block the model with the final check in Section~\ref{sec:check}.
\qed
\end{proof}

\begin{theorem}[Complexity Bound]
\label{thm:complexity}
Solving our particular encoding $\mathcal{E}_M$ is in the complexity class $\Sigma_2^P$ with respect to both the size of the input and the size of the matrix proof.
\end{theorem}
\begin{proof}
There are polynomially-many SAT variables.
To see this, let $c$ be the number of clauses in the input, containing a total of $l$ literals.
We have at most $d\cdot c$ selectors $S^k_C$.
We also have $O(d^2l^2)$ possible connection literals $\langle L \sim K\rangle$.
Hence, there are only polynomially-many instantiations of~\eqref{eq:matrixConnect}. After adding in the worst case all of them, the problem is in NP. We can non-deterministically guess an assignment for all polynomially many selectors and unification atoms. Checking the model can be done clearly in deterministic polynomial time.
Checking whether the model represents a matrix with a spanning set of connections is in co-NP. It can be solved by a separate SAT solver, which checks if the matrix $\sigma(M)$ represented by the SAT model is satisfiability. As we can solve $\mathcal{E}_M$ in NP with a co-NP oracle, the problem of solving our encoding for some fixed limit $d$ is in $\Sigma_2^P$.
\qed
\end{proof}
As checking the satisfiability of a set of clauses over rigid variables is  $\Sigma_2^P$-complete~\cite{DBLP:conf/stacs/Goubault94},  the complexity of our approach coincides with this theoretical bound.
\begin{corollary}[Termination]
A run for solving $\mathcal{E}_M$ at fixed $d$ terminates.
\end{corollary}

%% file: unsatCore.tex
A  downside of our  encoding $\mathcal{E}_M$, especially of its constraints from Section~\ref{sec:fully-connected}, is that we eagerly introduce and use selectors for clause instances that are not required.
If there is more than one input clause and the matrix is of size $d$, not all clauses can have $d$ copies in the matrix for arithmetic reasons.
Therefore, creating $d$ instances of each clause is overkill. This section addresses this challenge and improves iterative deepening via unsat cores, resulting in a refined encoding 
$\mathcal{E}_U$.  

We use an abstraction-refinement~\cite{cegar} approach to approximate the number of copies required for each clause.
This way, we avoid polluting the search space with likely-unnecessary clause instances.
Instead of a coarse global limit $d$, we estimate how many copies of each clause are required with a \emph{multiplicity} $\mu$~\cite{matings-in-matrices}. Initially we have $\mu(C) = 1$ for start clauses and $\mu(C) = 0$ otherwise.
The multiplicity is monotonically increased based on the unsat core of the following encoding.
We refine constraint~\eqref{eq:matrixConnect} to 
\begin{equation}
    \label{eq:matrixConnect2}
    S_C \Vdash~\bigvee_D~\bigvee_{1 \le k \le \bm{\mu(D) + 1}}~\bigvee_{K \in D^k} S^k_D \wedge \langle L \sim K\rangle
\end{equation}
as we  have $\mu(D)$ copies of $D$.
Note that $k$  ranges up to $\mu(D) + 1$.
We add temporary assertions\footnote{named $\kappa$ because it indicates that a clause needs more ``$\kappa$-city''
} $\kappa_D := \lnot S_D^{\mu(D) + 1}$ so that the solver cannot  select $D^{\mu(D) + 1}$, but it \emph{can} report that finding a proof failed in part due to a lack of copies of $D$.

We revise~\eqref{eq:openpath1}, as we can no longer assume that a fully connected matrix has exactly $d$ clauses. 
A candidate matrix can be fully connected, but the final proof may in fact have a matrix that is a \emph{superset} of the candidate.
As~\eqref{eq:openpath1} is now too strong, we weaken it to
\begin{equation}
\label{eq:openpath2}
\bar{S} \Vdash \left[\bigvee_{\{L, K\} \subseteq U} \langle L \sim K\rangle\right] \bm{\vee \bigvee_{L \in U} F_L}
\end{equation}
where $F_L$ is a formula indicating that $L$ could also be connected to another literal in a clause \emph{not yet in the matrix} and $\bar{S}$ a set of selectors as before in~\eqref{eq:openpath1}.
Whenever the SAT solver reports unsatisfiability, we retrieve the \emph{unsat core} representing a potentially non-minimal subset of $\kappa$ assertions sufficient to yield unsatisfiability.
We may increase one or more $\mu(C)$ if the corresponding assertion occurs in the unsat core.
However, to retain completeness we need to ensure that we eventually increment the multiplicity of every clause appearing repeatedly in the unsat core: in other words, we require \emph{fairness}.
In case the core is empty, we can conclude that no proof exists.
As a result, our SAT encoding $\mathcal{E}_U$ with improved iterative deepening is given by~\eqref{eq:start:matrix},~\eqref{eq:matrixConnect2}, and~\eqref{eq:openpath2}.
\begin{example}
\label{ex:infExt}
Consider the input problem
\begin{align*}
&C := P(a) &&D := \forall x.~\lnot P(x) \vee P(f(x)) &&E := \forall y.~\lnot P(y)
\end{align*}
with $C$ as start clause. $\kappa_{D}$ will always be contained within the unsat core, no matter its multiplicity. 
However, a fair enumeration eventually includes $\kappa_{E}$, and we find the obvious proof.
\end{example}
Our improved encoding $\mathcal{E}_U$ remains sound and terminating by similar arguments to Theorems~\ref{thm:sound}~and~\ref{thm:complexity}. Completeness requires an adjusted argument.
\begin{theorem}[Completeness]
If a matrix $M$ together with a spanning set of connections exist, there is a corresponding propositional model of $\mathcal{E}_U$.
\end{theorem}
\begin{proof}
In addition to Theorem~\ref{thm:complete}, we show that if there is a proof using $M$ which our current $\mu$ does not permit, at least one relevant $\kappa_C$ is contained in the unsat core.
Fairness then ensures we will eventually find the proof.
Consider a maximal subset $M' \subset M$ representable at $\mu$.
\begin{enumerate}[label=(\arabic*),wide=0em,leftmargin=0em]
\item If $M'$ cannot be fully connected, it contains at least one literal $L$ with no connections. Since $M'$ is maximal and $M$ can be fully connected, $L$ should be connected to some literal in a clause $D$ not yet in the matrix.
This option is offered in~\eqref{eq:matrixConnect2}, but fails
because the respective $\kappa_D$ assumption is forced false. $\kappa_D$ is therefore in the unsat core.
\item If $M'$ can be fully connected, we would have failed to close some open path $U$ and propagated some instance of~\eqref{eq:openpath2}.
Some $L \in U$ must connect to at least one literal of a clause not yet in $M'$ by the right disjunct of~\eqref{eq:openpath2}.
As $M'$ is maximal, we can add no clauses and so the constraint fails because of the $\kappa$ assumption.
\end{enumerate}
\qed
\end{proof}
A beneficial side effect of our SAT encoding $\mathcal{E}_U$ is it also terminates on some non-theorems.
In combination with techniques introduced in Section~\ref{sec:optimisations}, we obtain a decision procedure for the effectively-propositional fragment in Theorem~\ref{thm:epr}.

%% file: optimisations.tex
When solving the SAT encodings of Sections~\ref{sec:tableau}--\ref{sec:unsatCore}, restricting the SAT solver's search space is beneficial. In addition to standard techniques,  such as tautology elimination~\cite{handbook-ar-model-elimination}, we propose some specialised redundancy eliminations.

\subsection{Multiplicity Symmetry}\label{sec:sb}
Our encodings from Sections~\ref{sec:tableau}--\ref{sec:unsatCore} contain \emph{several symmetries}~\cite{DBLP:journals/tc/AloulSM06}, which we now \emph{avoid}, rather than \emph{break}~\cite{symmetry-breaking-fmb}.
The first symmetry is that copies of clauses are interchangeable.
Suppose we select connect some literal $L$ to literal $K$ in a copy of $C$ not yet in the matrix, and subsequently fail to find a proof in that direction.
Nothing prevents the SAT solver selecting another so-far-unused copy of $C$ and failing for virtually the same reasons as before.
We avoid this by propagating
\begin{equation}
    \label{eq:clauseIndexOrder}
    S^{i + 1}_C \Vdash S^i_C, 
\end{equation}
enforcing that $C^i$ can be selected only if all $C^j$ with $j < i$ are selected, eliminating this symmetry.

\subsection{Subsumption and Instance Symmetry}
Saturation systems often delete a clause $C$ because it is \emph{subsumed}~\cite{term-indexing} by some more-general clause $D$.
Dynamics in connection systems are somewhat different as new first-order clauses are not deduced, but nonetheless we can profit by applying some amount of subsumption.
If two different clauses $C$ and $D$ are in the current matrix, we can enforce that neither becomes a subset of the other, modulo $\sigma$\footnote{note that we do not apply an additional substitution to either side}.
This restriction preserves completeness, by Bibel's Lemma 6.8~\cite{Bibel}.

An obvious extension of this idea is to remove clauses from the matrix that are subsumed by other clauses from the input set. This, however, fails.
\begin{example}
    Consider the four input clauses
\begin{align*}
&C := P(a) &&D := Q(a)\\
&E := \forall x.~\lnot P(x) \vee Q(x) &&F := \forall y.~\lnot Q(y)
\end{align*}
with $C$ the only start clause.
There is a proof without subsumption via $C$, $E$, and finally $F$, and in fact this is the only minimal proof using $C$. However, putting $E$ in the matrix with $\sigma(x) = a$ results in it being subsumed by $D$ from the input.
\end{example}
Subsumption in the usual sense of smaller clauses representing any usage of larger clauses fails. As we saw, this is because we might lose the \emph{reason to connect} a clause to our current matrix. Keeping larger clauses instead also does not work, as we might not be able to connect all literals of the larger clause.
Nonetheless, we can motivate additional symmetry avoidance this way.
Define an arbitrary total order $\prec$ on input clauses such that start clauses are the least elements. We assume that the order of each clause in the matrix is the same as the order of the clauses in the input set from which they are a copy.

\begin{lemma}[Instance Symmetry]
Suppose there is a matrix $M$ with a spanning set of connections containing a clause $D$ with $D \succ C$, and that there is a $\rho$ such that $\rho(C) = \sigma(D)$.
Then $M$ with $D$ exchanged for $C$ also has a spanning set of connections.
\end{lemma}
\begin{proof}
As all variables in $C$ and $D$ are fresh, we can adapt $\sigma$ according to $\rho$. This way, $C$ may be connected to the same literals as $D$.
As $\rho(C)$ has the same literals as $\sigma(D)$, we neither add additional paths that must be closed, nor do we prevent other clauses connecting to $C$ because we dropped the respective literal.
\end{proof}
\begin{corollary}[Instance Symmetry Completeness]
Forbidding any such $D$ during search remains complete.
\end{corollary}

\subsection{Substitution Symmetry}
A related symmetry appears within the substitution applied to different copies of the same clause.
\begin{example}
Consider a literal in two copies of the same clause, $L[x]$ and $L[y]$. 
Assume that all attempts with $\sigma(x) = a$ and $\sigma(y) = b$ fail.
Nothing prevents trying again with all connections ``flipped'' to the other clause and $\sigma(x) = b$ and $\sigma(y) = a$, introducing an exponential number of branches in the worst case.
\end{example}
We enforce an ordering on substitution of \emph{variables in copies of the same clause}.
This ordering of terms should be stable under substitution and orient as many terms as possible, but need not have the subterm property and therefore may not be a reduction ordering~\cite{term-rewriting-and-all-that}.
We suggest the following order.

Assume an arbitrary total ordering $\prec$ over function symbols.
Define $f(\bar{t}) \prec g(\bar{s})$ iff (i) $f \prec g$ or (ii) $f = g$ and $\bar{t} \prec \bar{s}$.
Sequences of terms $\bar{t} \prec \bar{s}$ are compared lexicographically.
Now, let $\bar{x}$ be the variables occurring left-to-right in clause $C$.
Given two instances $C^i$ and $C^j$ of the same clause with $i < j$, we may enforce that $\sigma(\bar{x}_i) \nsucceq \sigma(\bar{x}_j)$ to avoid symmetries over clause substitutions.
\begin{lemma}[Spanning Order]\label{lem:reorder}
Suppose $M$ has a spanning set of connections and contains two copies $C^i$ and $C^j$ of the same clause.
Then there is a spanning set of connections that satisfies $\sigma(\bar{x}_i) \nsucceq \sigma(\bar{x}_j)$.
\end{lemma}
\begin{proof}
If this condition does not already hold, we have $\sigma(\bar{x}_i) \succeq \sigma(\bar{x}_j)$.
Duplicate clauses are already eliminated, so in fact $\sigma(\bar{x}_i) \succ \sigma(\bar{x}_j)$.
Now ``swap'' $C^i$ and $C^j$ by exchanging their connections to obtain a new spanning set of connections and consistent substitution $\sigma'$.
Necessarily, $\sigma'(\bar{x}_i) \prec \sigma'(\bar{x}_j)$.
\qed
\end{proof}
By iterated application of Lemma~\ref{lem:reorder}, it is possible to ``reorder'' any spanning set of connections into another that respects the order.
\begin{corollary}[Substitution Symmetry Completeness]
Enforcing an ordering on substitution of variables in copies of the same clause remains complete.
\end{corollary}

\subsection{Relating Unification, Constraints, and the Herbrand Universe}
\label{sec:herbrand}
Classically, connection systems maintain a substitution $\sigma$ and periodically check a set of constraints individually, backtracking if any constraint fails~\cite{handbook-ar-model-elimination}.
While efficient, this approach does not take into account \emph{mutually} unsatisfiable constraints, or the Herbrand universe.
For example, given there are only two constants $a, b$ in the universe, the set of constraints $x \neq a, x \neq b$ are individually satisfiable, but not together.
Ordering constraints also produce this effect: consider $x \prec y, y \prec z, z \prec x$.
Or, consider the universe generated by a constant $a$ and a unary function $f \succ a$.
Here, $x \prec a$ is unsatisfiable and $x \prec f(a)$ implies $x = a$.
There is a tradeoff between the pruning effect of such interrelated constraints and the computation required to enforce them, which must be considered for any future practical implementation.

\subsection{Clause Splitting}
\emph{Clause splitting} is a powerful technique in saturation-based theorem proving, decomposing clauses $C$ into variable-disjoint \emph{components} $C_1 \vee \ldots \vee C_n$ and dispatching them separately~\cite{superposition-sorts-splitting}.
Splitting is also of interest in analytic tableaux~\cite{method-of-variable-splitting} and the connection method~\cite{Bibel}.
We propose a splitting method specialised to our setting based on the AVATAR framework~\cite{AVATAR}.
AVATAR introduces a SAT variable $\alpha_i$ for each component $C_i$ of a clause and adds their disjunction to a SAT solver.
If this solver yields a model assigning $\alpha_i$ true, the component $C_i$ may be used as if it were in the input set.
When a first-order refutation is found, AVATAR adds a clause blocking the combination of $\alpha$ variables whose components were used in the refutation.
This process is repeated until the set of constraints becomes unsatisfiable.

Integrating clause splitting into connection systems requires special attention. 
Clauses in the matrix may become splittable at some point modulo $\sigma$, but on backtracking are no longer.
We therefore observe and record all splittable clause instances generated at some point during a run, but this requires a restart.

\begin{example}[Necessity of Restarts]
In all our encodings,
the set of possible connections must be known in advance. Adding components to the input -- and therefore possible connections -- after we have already propagated SAT clauses does not work properly.
Assume we already propagated an instance of~\eqref{eq:matrixConnect} and we later add a new component to the clause set that contains a possible connection. Adding~\eqref{eq:matrixConnect} again results in a strictly weaker and thus redundant constraint.
\end{example}
This is not a major problem as we can add new components each time we start the SAT solving process afresh, such as when the resource limit is increased after finding no proof.
However, adding components to the input without also excluding appropriate instances of the parent clause introduces duplication into search, and excluding parent clauses runs into more trouble.
\begin{example}[Connection Problems]
Consider clauses $P(x)$, $\lnot P(a) \vee Q(a)$, and $\lnot Q(x)$. Suppose $P(x)$ is the start clause, and we split the binary clause and remove it. Finding a sub-proof for the remaining input with $\lnot P(a)$ is straightforward, but for $Q(a)$ we need the original binary clause to make the connection.
\end{example}
This is why we suggest a different approach that does not add components explicitly, but instead ``relaxes'' some literals in a clause.
Let $C$ be a clause (modulo $\sigma$ in general) such that $C$ is splittable into variable-disjoint components $C_i$.
Define the \emph{active literals in $C$} to be the union of all $C_i$ such that $\alpha_i$ is assigned true in the AVATAR model.
If $C$ is instead not splittable, all its literals are defined to be active.
Note that a literal can be active in one clause and inactive in another, even if they are copies of the same input clause.
For the sake of simplicity, assume encoding $\mathcal{E}_M$.
We now relax constraints on literals that are not active, so that~\eqref{eq:matrixConnect} becomes
\begin{equation}
\label{eq:avatarConnect}
S_C, \bm{A^C_L} \Vdash~\bigvee_D~\bigvee_{1 \le k \le d}\bigvee_{K \in D^k} S^k_D \wedge (\bm{A^D_K \Rightarrow} \langle L \sim K\rangle).
\end{equation}
where $A^C_L$ is an SAT variable expressing that literal $L$ is active in clause $C$.
The assignment of such variables can be checked internally with respect to $\sigma$ and the AVATAR model, in a similar way to unification constraints.

We also relax the definition of spanning set of connections to ignore open paths if said path contains inactive literals.
In case we find such a spanning set of connections, we block the corresponding set of AVATAR variables with a conflict clause over $\alpha_i$ and obtain a new model.
Note that~\eqref{eq:avatarConnect} only requires that active literals need to be connected. This avoids the previously discussed problem that we are not able to connect to some clauses because AVATAR selected only parts of it. We still add clauses containing literals that \emph{could} be connected, but we are not required to actually perform connections to inactive literals.

The relaxed encoding remains sound as the resulting set of connections has no open paths through active literals, i.e. it is spanning for some matrix of input clauses and components assigned true in the AVATAR model.
Completeness can be obtained immediately by noticing that the encoding is strictly weaker than $\mathcal{E}_M$ for any given AVATAR model.
When the space of AVATAR models is eventually exhausted, a proof can be given consisting of multiple matrix sub-proofs.

%% file: epr.tex
Assume that we have at least the improved iterative deepening technique of Section~\ref{sec:unsatCore} querying unsat cores, disallow duplicate clauses, and have a sufficiently-powerful implementation of Section~\ref{sec:herbrand} to reason about finite Herbrand universes.
Then, search for solutions to the encoding becomes a decision procedure for the effectively propositional fragment (EPR)~\cite{epr}.

\begin{theorem}[EPR decidability]
\label{thm:epr}
Assuming effectively propositional input, proof search over $\mathcal{E}_U$ terminates.
\end{theorem}
\begin{proof}
By definition, all symbols in the input clauses are constants.
If the input is a theorem, the procedure terminates by completeness.
Therefore, suppose the input is not a theorem, and so each run of the solver terminates with unsatisfiability.
It suffices to show that the unsatisfiable core of Section~\ref{sec:unsatCore} will eventually become empty, indicating that the input is not a theorem.

Let $c$ be the number of constants in the Herbrand universe.
There are at most $c^v$ possible instantiations of a clause $C$, where $v$ is the number of variables in $C$.
Assume the limit $\mu(C)$ of this clause has reached $c^v + 1$.
Choosing all available selector variables $S^k_C$ would conflict either with constraint~\eqref{eq:clauseIndexOrder} or with the requirement that all clause copies are distinct.
$\kappa_C$ will therefore not appear in a minimal unsat core, as it can be shown false independently of the assumption.
Consequently, $\mu(C)$ will not increase further.
Every clause will eventually reach their limit and will not occur in the unsat core from that point onwards, and eventually the core becomes empty.
\qed
\end{proof}

%% file: related.tex
First-order theorem provers employ a variety of ground reasoning techniques, predominantly SAT and SMT solvers.
Here we must mention the family of instance-based methods~\cite{instance-based-methods}: grounding a set of first-order clauses in the hope that they become unsatisfiable, which can be employed with a dedicated calculus~\cite{InstGen} or alongside an existing system~\cite{SATCoP,e-grounding}.
In the other direction, SMT solvers often integrate quantifier instantiation into satisfiability routines~\cite{mbqi,e-matching}.
Ground reasoning can also be used for many other combinatorial tasks in first-order theorem provers~\cite{uses-of-sat-in-vampire}, such as keeping track of clause splitting~\cite{AVATAR}, detecting subsumption~\cite{sat-subsumption}, or determining when inferences are applicable~\cite{sat-subsumption-resolution}.
MACE-style finite model builders~\cite{mace} employ SAT solving to determine whether a set of clauses is satisfiable assuming a finite model of fixed size, and symmetry-breaking can also be applied~\cite{symmetry-breaking-fmb}.

Restricting ourselves now to directly encoding proof objects, the ChewTPTP system in both its SAT~\cite{chewtptp-sat} and SMT~\cite{chewtptp-smt} incarnations is the closest existing approach to theorem proving via satisfiability.
ChewTPTP encodes constraints for a closed connection tableau completely ahead of time, then passes the resulting constraints to a SAT or SMT solver.
We have ourselves previously published an early version of our ideas in a more general setting~\cite{upCoP}.

%% file: conclusion.tex
We  encode first-order connection calculus as a propositional problem. We improve our SAT encodings for  matrix forms and by guiding iterative deepening using unsat cores. Furthermore, we discuss several optimizations to prune symmetries and eliminate unnecessary branches. Implementation and practical experimentation with our SAT-based approach is left for future work.